\documentclass[preprint,1p,12pt]{elsarticle}

\usepackage{amsmath, amsthm, amsfonts, amssymb}
\usepackage{changepage}
\usepackage{graphicx,subcaption,color,tikz}

\usepackage{pdfpages, fullpage}

\usepackage[linesnumbered,ruled,vlined]{algorithm2e}

\newtheorem{thm}{Theorem}

\newtheorem{cor}{Corollary}

\newtheorem*{rem}{Remark}

\journal{Information Processing Letters}

\begin{document}

\begin{frontmatter}

%% Title, authors and addresses

%% use the tnoteref command within \title for footnotes;
%% use the tnotetext command for theassociated footnote;
%% use the fnref command within \author or \address for footnotes;
%% use the fntext command for theassociated footnote;
%% use the corref command within \author for corresponding author footnotes;
%% use the cortext command for theassociated footnote;
%% use the ead command for the email address,
%% and the form \ead[url] for the home page:
%% \title{Title\tnoteref{label1}}
%% \tnotetext[label1]{}
%% \author{Name\corref{cor1}\fnref{label2}}
%% \ead{email address}
%% \ead[url]{home page}
%% \fntext[label2]{}
%% \cortext[cor1]{}
%% \address{Address\fnref{label3}}
%% \fntext[label3]{}

%% use optional labels to link authors explicitly to addresses:
%% \author[label1,label2]{}
%% \address[label1]{}
%% \address[label2]{}

\title{Amortized Rotation Cost in AVL Trees}

\author[auth1]{Mahdi Amani} %\fnref{cor2}}
%\fntext[cor2]{Work partially supported by the Italian Ministry of Education, University, and Research (MIUR) under PRIN 2012C4E3KT national research project AMANDA --- Algorithmics for Massive and Networked Data.}
\author[auth2]{Kevin A. Lai\corref{cor1}}
\ead{kalai@princeton.edu}
\author[auth2,auth3]{Robert E. Tarjan}
%\date{\today}
\cortext[cor1]{Corresponding author}

\address[auth1]{Universit\`a di Pisa, Dipartimento di Informatica}
	%Email: m\_amani@di.unipi.it}
\address[auth2]{Princeton University, Computer Science Department}
	%Email: ret@cs.princeton.edu}
\address[auth3]{Intertrust Technologies, Sunnyvale, CA}

%\begin{abstract}

%\end{abstract}

\begin{keyword}
data structures \sep AVL tree\sep lower bound \sep rotation

%% keywords here, in the form: keyword \sep keyword

%% PACS codes here, in the form: \PACS code \sep code

%% MSC codes here, in the form: \MSC code \sep code
%% or \MSC[2008] code \sep code (2000 is the default)

\end{keyword}

\end{frontmatter}

\section{Introduction}
An AVL tree \cite{adelson} is the original type of balanced binary search tree.  An insertion in an $n$-node AVL tree takes at most two rotations, but a deletion in an $n$-node AVL tree can take $\Theta(\log n)$.  A natural question is whether deletions can take many rotations not only in the worst case but in the amortized case as well.  A sequence of $n$ successive deletions in an $n$-node tree takes $O(n)$ rotations \cite{tsakalidis}, but what happens when insertions are intermixed with deletions?  

Heaupler, Sen, and Tarjan \cite{haeupler} conjectured that alternating insertions and deletions in an $n$-node AVL tree can cause each deletion to do $\Omega(\log n)$ rotations, but they provided no construction to justify their claim.  We provide such a construction: we show that, for infinitely many $n$, there is a set $E$ of {\it expensive} $n$-node AVL trees with the property that, given any tree in $E$, deleting a certain leaf and then reinserting it produces a tree in $E$, with the deletion having done $\Theta(\log n)$ rotations. One can do an arbitrary number of such expensive deletion-insertion pairs. The difficulty in obtaining such a construction is that in general the tree produced by an expensive deletion-insertion pair is not the original tree. Indeed, if the trees in $E$ have even height $k$, $2^{k/2}$ deletion-insertion pairs are required to reproduce the original tree.

\section{Definition and Rebalancing of AVL Trees}
To define AVL trees, we use the rank-balance framework of Haeupler, Sen, and Tarjan \cite{haeupler}.  Although this gives a non-standard definition of AVL trees, it is equivalent to the original, and it is easier to work with.  A node in a binary tree is {\it binary}, {\it unary}, or a {\it leaf} if it has two, one, or no children, respectively.  A unary node or leaf has one or two {\it missing children}, respectively.  A {\it ranked binary tree} is a binary tree in which each node $x$ has a non-negative integer {\it rank} $x.r$.  By convention, a missing node has rank $-1$. The rank of a ranked binary tree is the rank of its root.  We denote the parent of a node $x$ by $x.p$.  The {\it rank difference} of a child $x$ is $x.p.r - x.r$.  A child of rank difference $i$ is an $i$-child; a node whose children have rank differences $i$ and $j$ with $i \le j$ is an {\it $i,j$ node}.  

An {\it AVL tree} is a ranked binary tree satisfying the following the {\it rank rule}: every node is 1,1 or 1,2.  Since missing nodes have rank $-1$, every leaf in an AVL tree is 1,1 and has rank 0, and every unary node is 1,2 and has rank 1.  Since all rank differences are positive, leaves have rank 0, and every node has a child of rank difference 1, we see that the rank of a node in an AVL tree equals its height.

Insertions and deletions in AVL trees can violate the rank rule.  We restore the rank rule by changing the ranks of certain nodes and doing rotations, which change the tree structure locally while preserving the symmetric order of nodes.  Figure \ref{rot} illustrates a rotation.

\begin{figure}[!hbtp]
\centering
\includegraphics[width=5in]{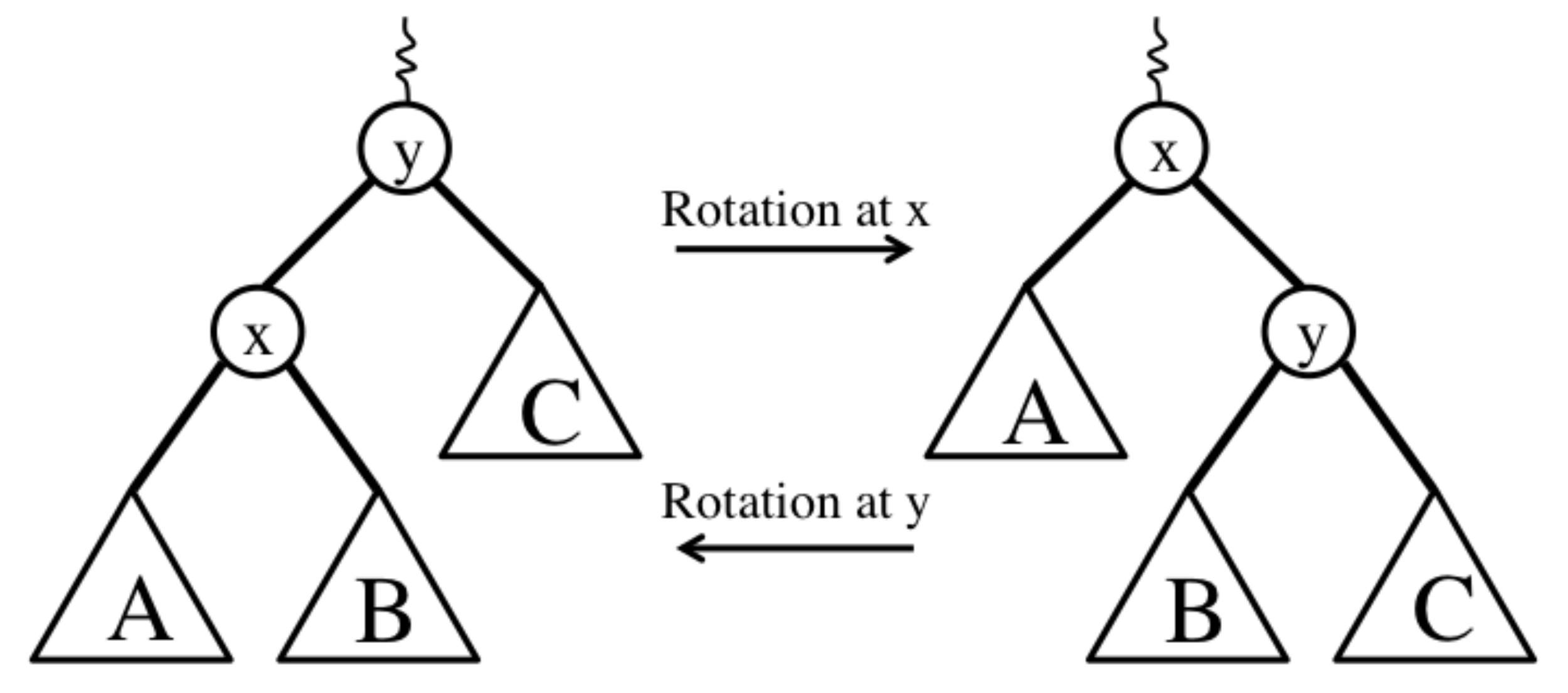}
\caption{Right rotation at node $x$. Triangles denote subtrees. The inverse operation is a left rotation at $y$.}
\label{rot}
\end{figure}

AVL trees grow by leaf insertions and shrink by deletions of leaves and unary nodes.  To add a leaf to an AVL tree, replace a missing node by the new leaf and give the new leaf a rank of 0.  If the parent of the new leaf was itself a leaf, it is now a 0,1 (unary) node, violating the rank rule.  In this case, rebalance the tree by repeatedly applying the appropriate case in Figure \ref{ins} until the rank rule holds. 

\begin{figure}[!hbtp]
\centering

\begin{subfigure}{\textwidth}
\centering
\includegraphics[width=5in]{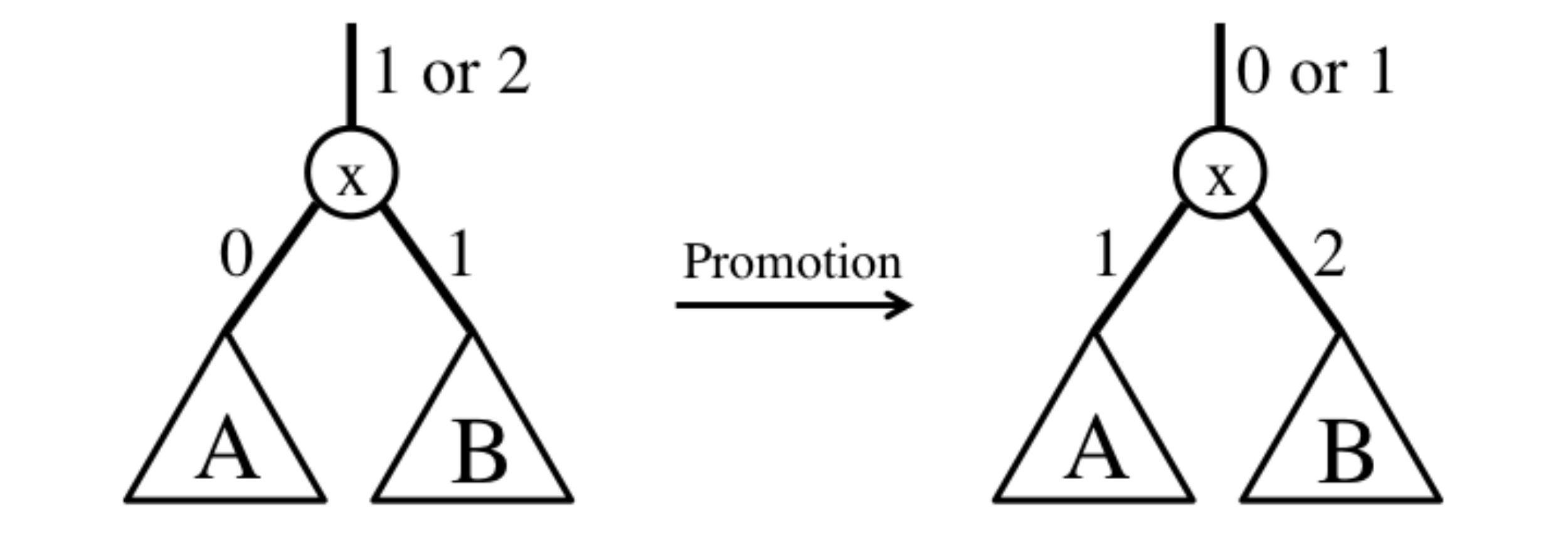}
\caption{Promotion to rebalance after insertion}
\label{ins1}
\end{subfigure}

\begin{subfigure}{\textwidth}
\centering
\includegraphics[width=5in]{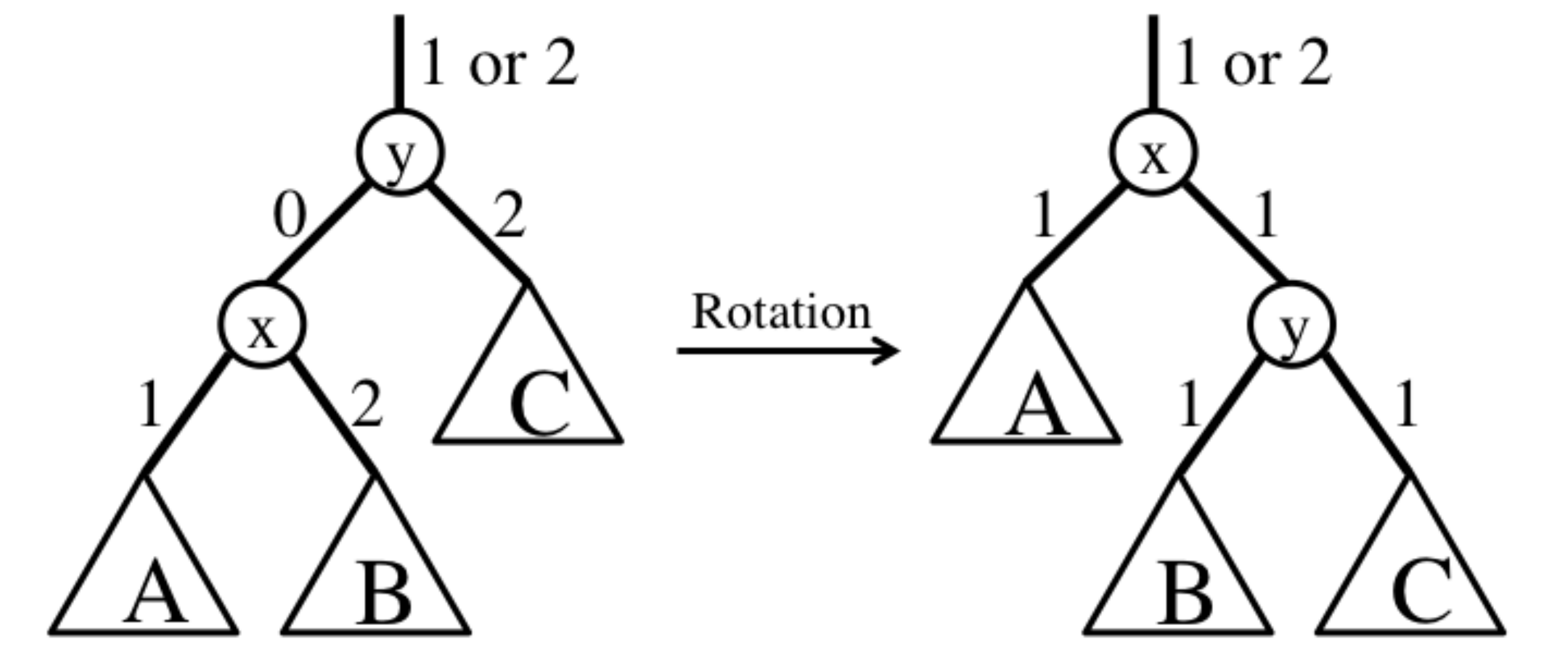}
\caption{Single rotation to rebalance after insertion}
\label{ins2}
\end{subfigure}

\begin{subfigure}{\textwidth}
\centering
\includegraphics[width=5in]{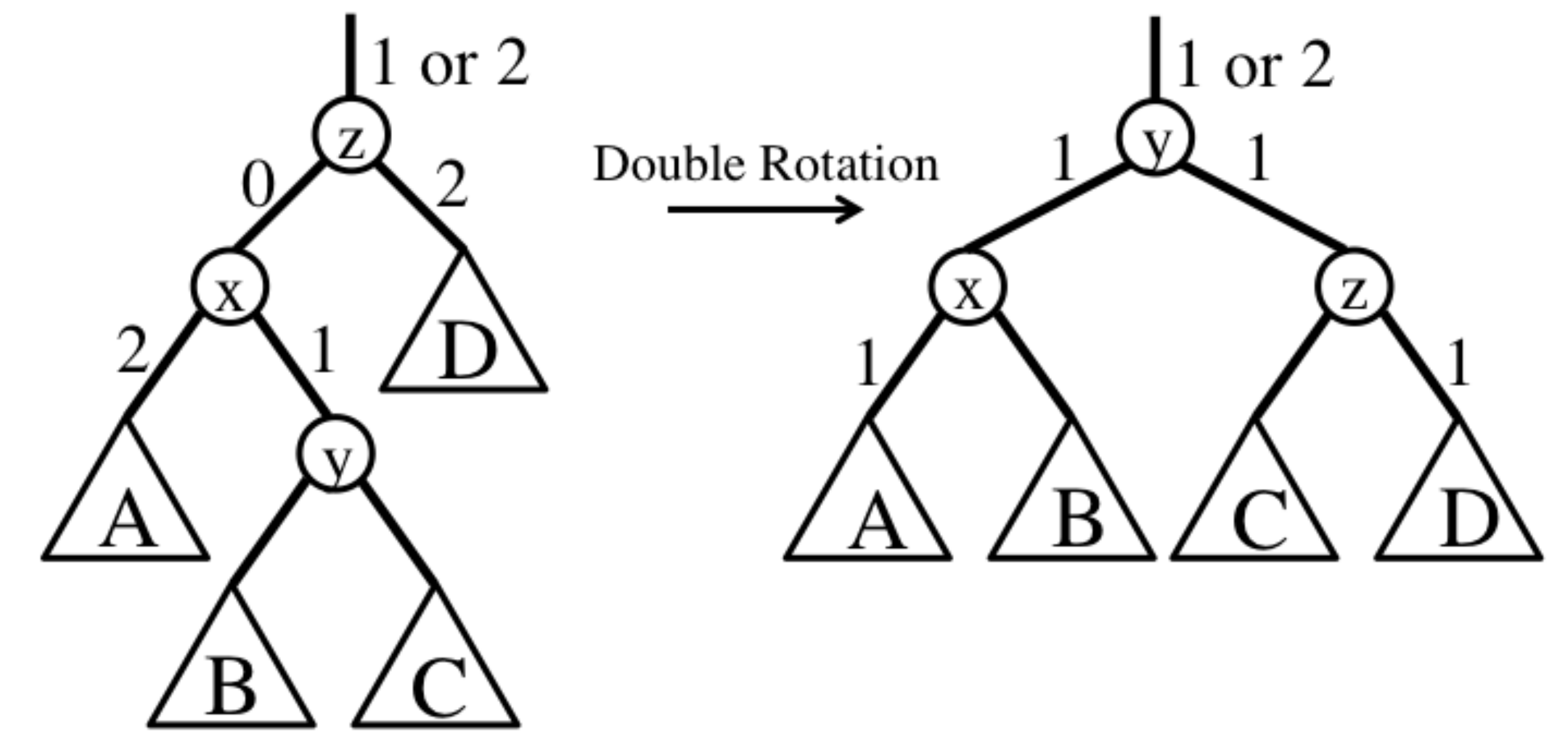}
\caption{Double rotation to rebalance after insertion}
\label{ins3}
\end{subfigure}

\caption{Rebalancing cases after insertion. Numbers next to edges are rank differences. Rank differences of unmarked edges do not change. The promote step may repeat. All cases have mirror images.}
\label{ins}
\end{figure}

A promotion (Figure \ref{ins1}) increases the rank of a node ($x$ in Figure \ref{ins1}) by 1. We call the node whose rank increases the {\it promoted node}. Each promotion either creates a new violation at the parent of the promoted node or restores the rank rule and terminates rebalancing.  Each single or double rotation (Figures \ref{ins2} and \ref{ins3}, respectively) restores the rank rule and terminates rebalancing.

To delete a leaf in an AVL tree, replace it by a missing node; to delete a unary node, replace it by its only child (initially changing no ranks).\footnote{Our expensive examples only delete leaves.  To delete a binary node $x$, swap $x$ with its symmetric-order successor or predecessor and proceed as described in the text; the swap makes $x$ a leaf or unary node.}  Such a deletion can violate the rank rule by producing a 2,2 or 1,3 node.  In this case, rebalance the tree by applying the appropriate case in Figure \ref{del} until there is no violation.  Each application of a case in Figure \ref{del} either restores the rank rule or creates a new violation at the parent of the previously violating node.  Whereas each rotation case in insertion terminates rebalancing, the rotation cases in deletion can be non-terminating.

\begin{figure}[!hbtp]
\centering

\begin{subfigure}{\textwidth}
\centering
\includegraphics[width=5in]{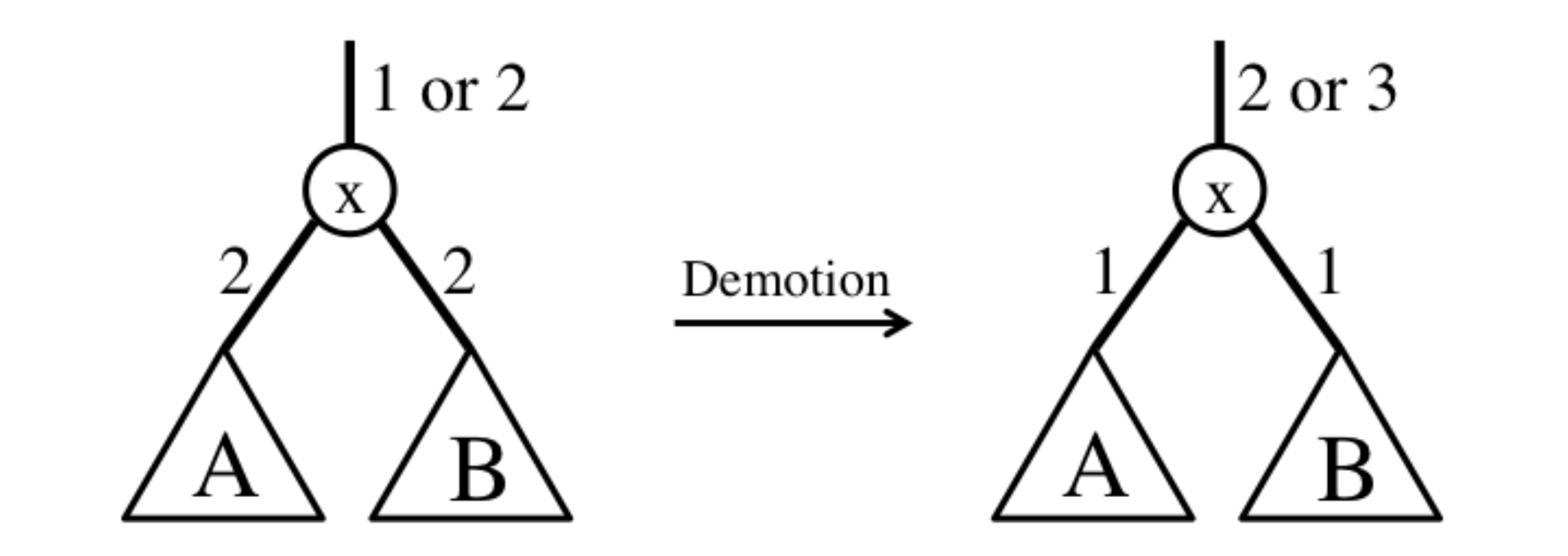}
\caption{Demotion to rebalance after deletion}
\label{del1}
\end{subfigure}

\begin{subfigure}{\textwidth}
\centering
\includegraphics[width=5in]{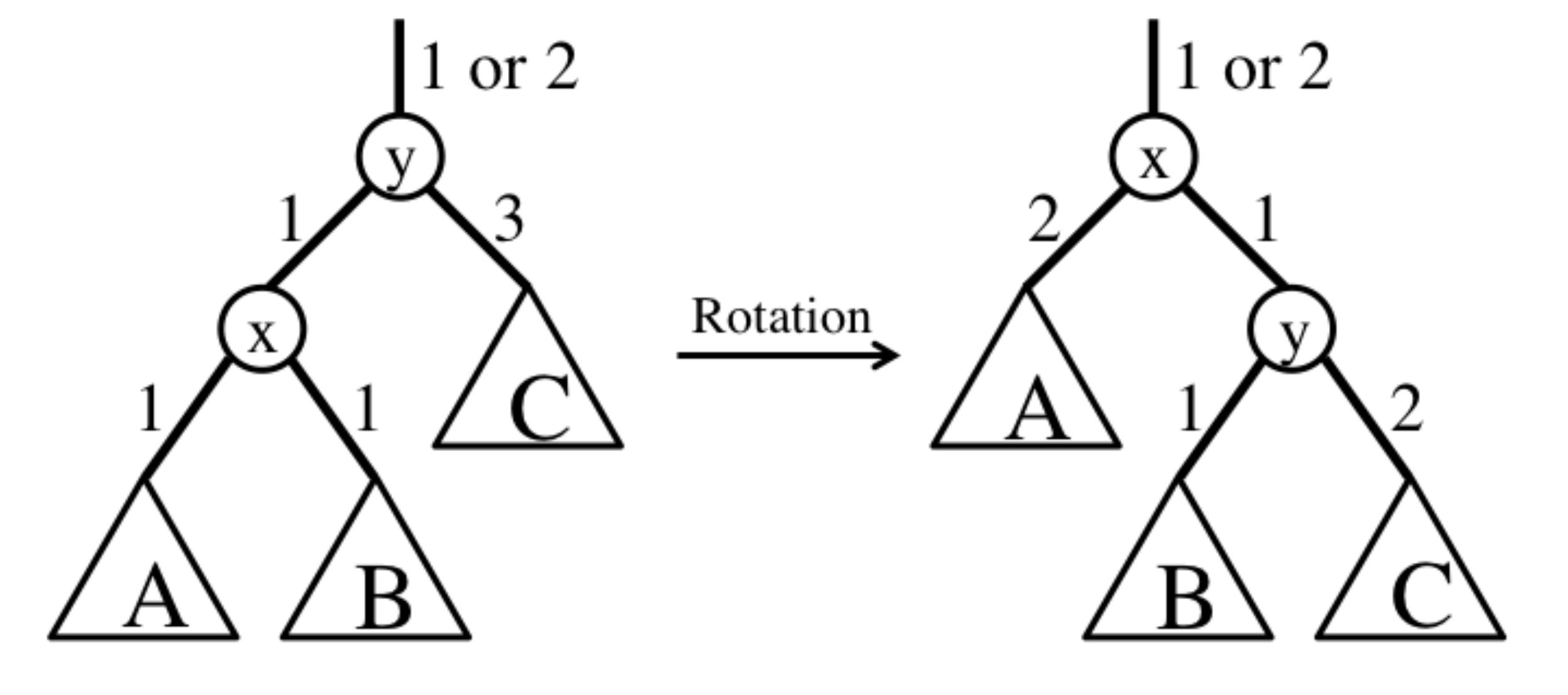}
\includegraphics[width=5in]{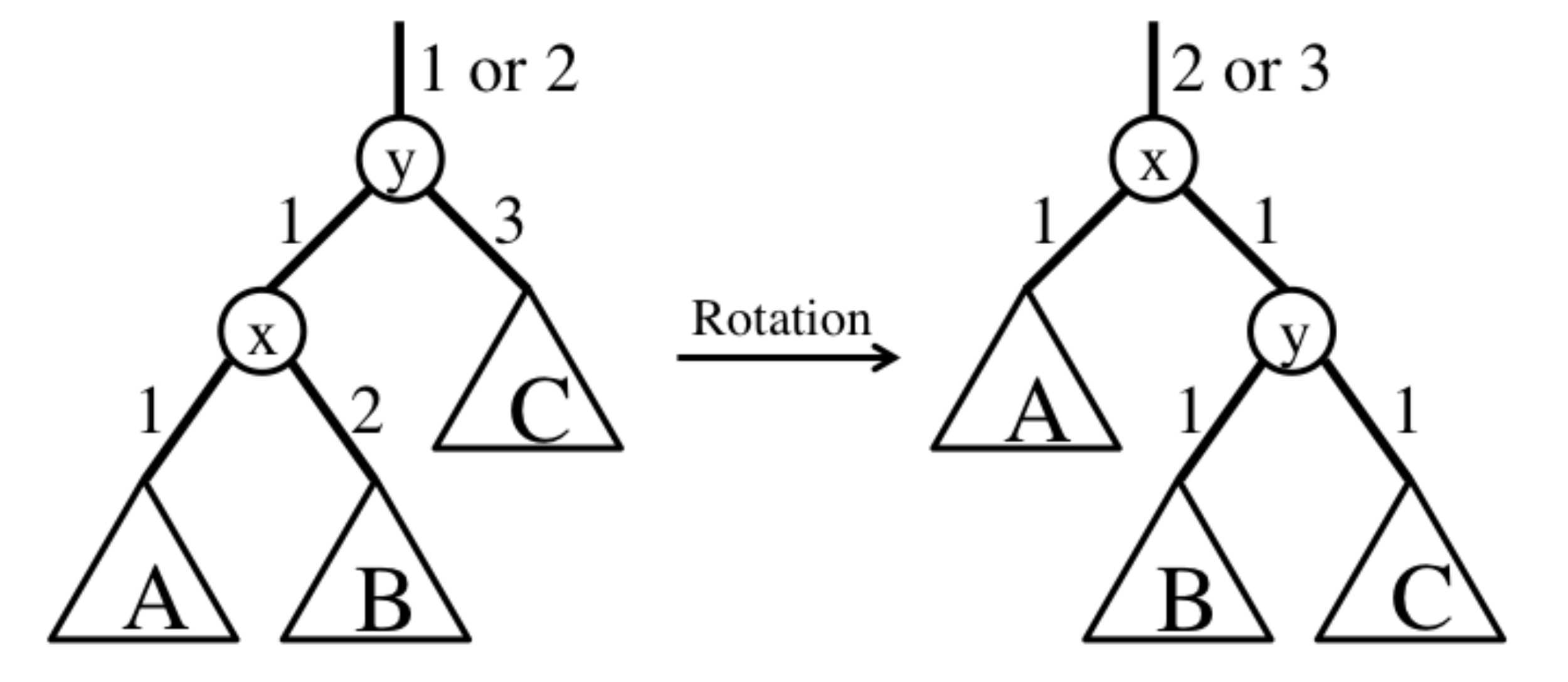}

\caption{Single rotation to rebalance after deletion}
\label{del2}
\end{subfigure}

\end{figure}
\clearpage
\begin{figure}
\ContinuedFloat
\begin{subfigure}{\textwidth}
\centering
\includegraphics[width=5in]{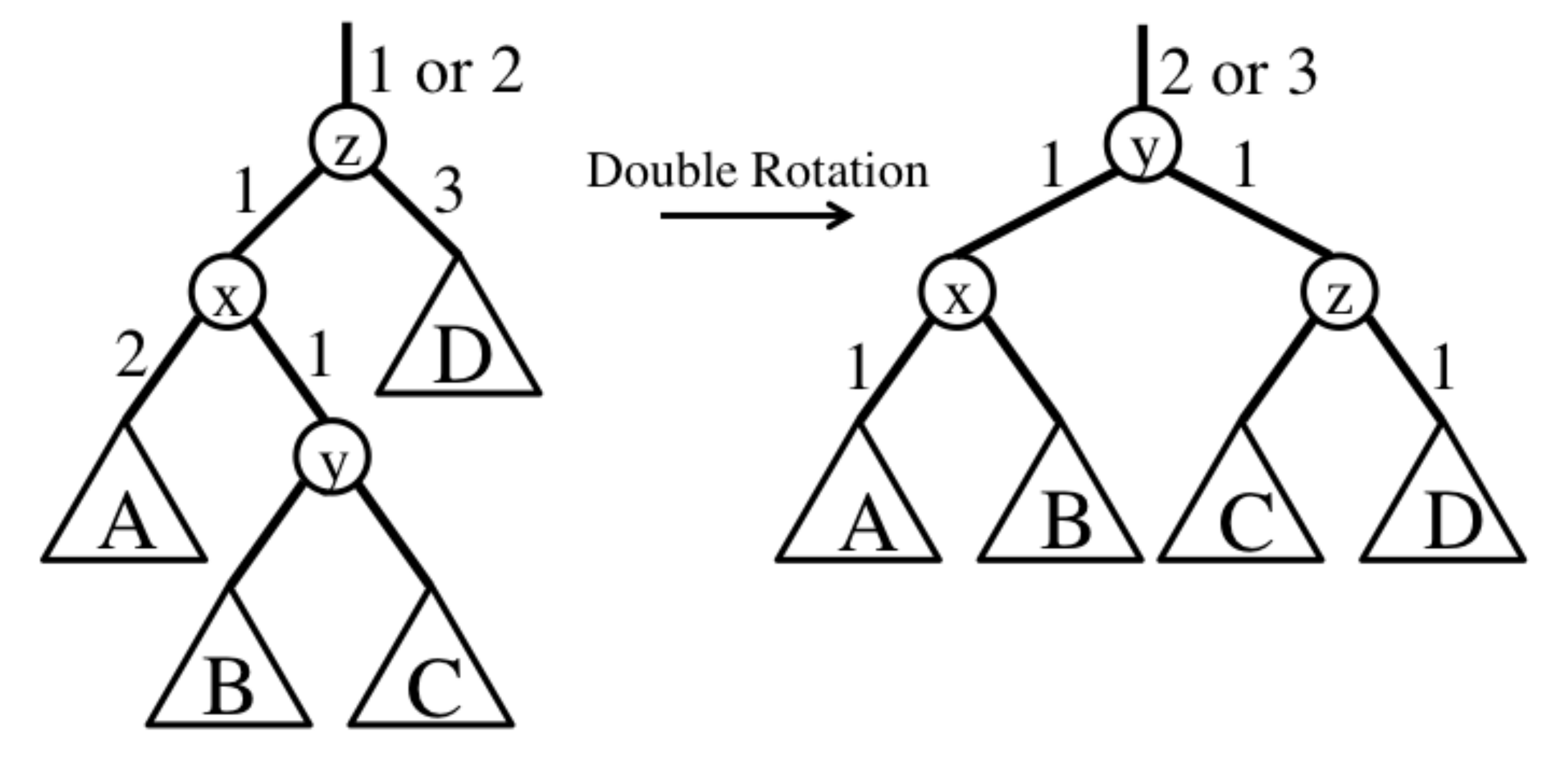}
\caption{Double rotation to rebalance after deletion}
\label{del3}
\end{subfigure}

\caption{Rebalancing cases after deletion. Numbers next to edges are rank differences. Rank differences of unmarked edges do not change. Each case except the first single rotation case may repeat. All cases have mirror images.}
\label{del}
\end{figure}

\section{Construction of AVL Trees}
In order to obtain an initial tree in our expensive set $E$, we must build it from an empty tree.  Thus the first step in our construction is to show that any $n$-node AVL tree can be built from an empty tree by doing $n$ insertions.  Although this result is easy to prove, we have not seen it in print before.\footnote{It also happens to be false for more relaxed types of balanced trees, such as weak AVL (wavl) trees \cite{haeupler}.  Not all $n$-node wavl trees can be built from an empty tree by doing insertions only; many require a number of intermixed insertions and deletions exponential in $n$.  This follows from an analysis using an exponential potential function like those in \cite{haeupler}.}

\begin{thm}\label{thm:build}
Any $n$-node AVL tree can be built from an empty tree by doing $T$ insertions, each of which does only promotions.
\end{thm}

\begin{proof}
Let $T$ be a non-empty AVL tree.  The {\it truncation} $\underline{T}$ of $T$ is obtained by deleting all the leaves of $T$ and decreasing the rank of each remaining node by 1.  We prove by induction on the rank $k$ of $T$ that we can convert its truncation $\underline{T}$ into $T$ by inserting the leaves deleted from $T$ to form $\underline{T}$, in an order such that each insertion does only promotions.  The theorem then follows by induction on the height of the desired tree.

The empty tree can be converted into the one-node AVL tree by doing a single insertion.  Thus the result holds for $k = 0$.  Suppose $k > 0$ and the result holds for any rank less than $k$.  Let $T$ be an AVL tree of rank $k$.  Tree $T$ consists of a root $x$ and left and right subtrees  $L$ and $R$, both of which are AVL trees.  The truncation $\underline{T}$ of $T$ consists of root $x$, now of rank $k - 1$, and left and right subtrees $\underline{L}$ and $\underline{R}$.  Both $L$ and $R$ have rank $k - 1$ or $k - 2$, and at least one of them has rank $k - 1$.  Suppose $R$ has rank $k - 1$.  By the induction hypothesis, $\underline{L}$ can be converted into $L$ and $\underline{R}$ can be converted into $R$ by inserting leaves, each insertion doing only promotions.  Out of these insertions into either $\underline{L}$ or $\underline{R}$, exactly one of them will increase the rank of the root by 1.

In the left subtree of $\underline{T}$, do the sequence of insertions that converts $\underline{L}$ into $L$.  Then, in the right subtree of the resulting tree, do the sequence of insertions that converts $\underline{R}$ into $R$.  If $L$ has rank $k - 1$, then the insertion into $\underline{L}$ that increases the root rank by 1 will, when done in $\underline{T}$, also increase the root rank of $\underline{T}$ by 1, from $k - 1$ to $k$, increasing the rank difference of the right child of the root from 1 to 2 but having no other effect on the right subtree of the root.  Thus, after all the insertions into the left subtree, the tree consists of root $x$, now of rank $k$, left subtree $L$, and right subtree $\underline{R}$ of rank $k- 2$.  The subsequent insertions into the right subtree will convert it into $R$ without affecting the rest of the tree, producing $T$ as the final tree.

If on the other hand $L$ has rank $k- 2$, then the insertions into the left subtree of $\underline{T}$ will convert the left subtree into L, in the process increasing the rank of the root of the left subtree from $k- 3$ to $k- 2$ but having no effect on the root or the right subtree.  The subsequent insertions will convert the right subtree into $R$.  Among these insertions, the one that increases the rank of the root of the right subtree from $k - 2$ to $k - 1$ will also increase the rank of $x$ from $k- 1$ to $k$, thereby converting the root of the left subtree from a 1-child to a 2-child but having no other effect on the left subtree.  Thus the final tree is $T$.  The argument is symmetric if $R$ has rank $k - 2$. 
\end{proof}         

\section{Expensive AVL Trees}

Our expensive trees have even rank.\footnote{It is easy to define an analogous set of expensive trees of odd rank.}  We define the set $E$ of expensive trees recursively.  Set $E$ is the smallest set containing the one-node tree of rank 0 and such that if $A$, $B$, and $C$ are AVL trees of rank $k$ such that $A$ and $C$ are in $E$, then the two trees of rank $k + 2$ shown in Figure \ref{E} are in $E$.  The tree of type $L$ in Figure 4 contains a root $x$ of rank $k + 2$ and a left child $y$ of the root of rank $k + 1$, and has $A, B$, and $C$ as the left and right subtrees of $y$ and the right subtree of $x$, respectively.  The tree of type $R$ in Figure \ref{E} is similar except that $x$ is the right child of $y$ and $A, B,$ and $C$ are the left subtree of $y$ and the left and right subtrees of $x$, respectively.

\begin{figure}
\centering
\includegraphics[width=5in]{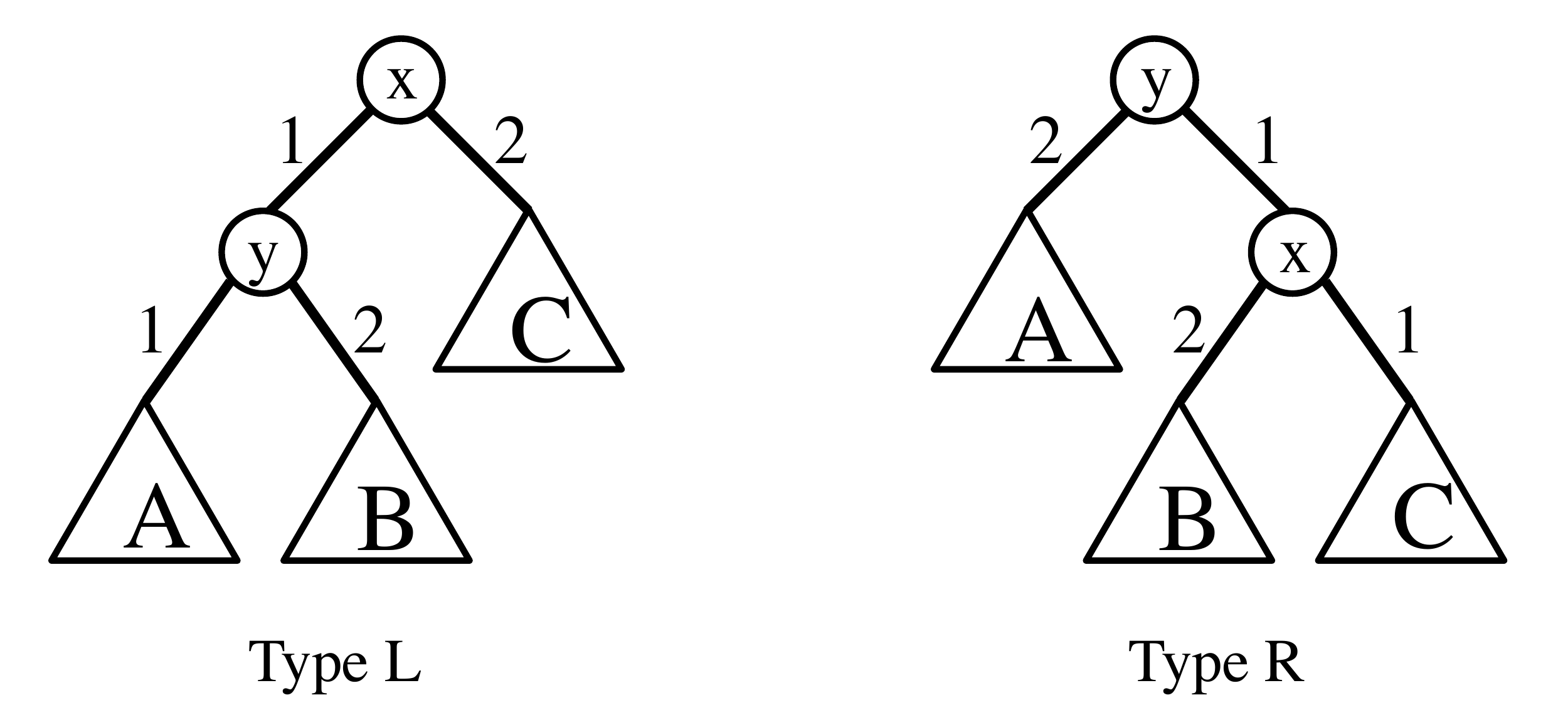}

\caption{Recursive definition of $E$. Numbers on edges are rank differences. The two trees shown are in $E$ if $A$ and $C$ are in $E$, $B$ is an AVL tree, and $A$, $B$, and $C$ have the same rank.}
\label{E}
\end{figure}

If $T$ is a tree in $E$, its shallow leaf is the leaf $z$ such that all nodes on the path from $z$ to the root, except the root itself, are 2-children.  A straightforward proof by induction shows that the shallow leaf exists and is unique.

\begin{thm}\label{thm:cost}
If $T$ is a tree in $E$ of rank k, deletion of its shallow leaf takes k/2 single rotations and produces a tree of rank $k - 1$. Reinsertion of the deleted leaf takes $k$ promotions and produces a tree of rank $k$ that is in $E$.
\end{thm}

\begin{proof}
We prove the theorem by induction on $k$.  In the one-node tree of rank 0, the shallow leaf is the only node.  Its deletion takes no rotations and produces the empty tree; its reinsertion takes no promotions and reproduces the original tree.  For $k = 2$, there is exactly one tree in $E$ of type $L$ and one of type $R$.  As shown in Figure \ref{base}, rebalancing after deletion of the shallow leaf in the type-$L$ tree takes one rotation and produces a tree of rank 1, and reinsertion takes two promotions and produces the type-$R$ tree.  Symmetrically, deletion of the shallow leaf in the type-$R$ tree takes one rotation and produces a tree of rank 1, and reinsertion takes one promotion and produces the type-$L$ tree.

\begin{figure}
\centering
\includegraphics[width=5.5in]{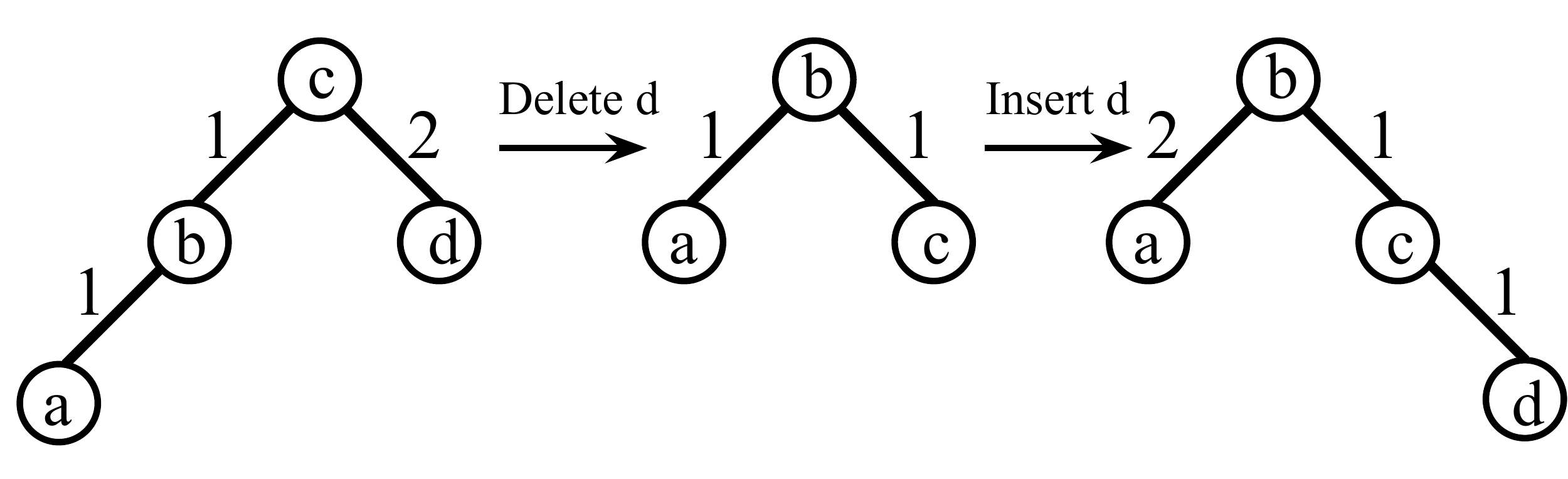}

\caption{Deletion and insertion of the shallow leaf in a type-$L$ tree of rank 2.}
\label{base}
\end{figure}

Suppose the theorem is true for $k$.  Let $T$ be a tree of rank $k$ + 2 and type $L$ in $E$.  (The argument is symmetric for a tree of type R.)  Let $x$ be the root, $y$ the left child of $x$, and $A$, $B$, and $C$ the left and right subtrees of $y$ and the right subtree of $x$, respectively (See the first tree in Figure \ref{ind}).  The shallow leaf of $C$ is the shallow leaf of $T$.  By the induction hypothesis, its deletion in $C$ does $k/2$ rotations and converts $C$ into a tree $C'$ of rank $k - 1$.  In $T$, deletion of the shallow leaf converts the right subtree of $x$ into $C'$, making the root of $C'$ a 3-child (See the second tree in Figure \ref{ind}). This causes one more single rotation, for a total of $k/2 + 1$, and produces the tree $T'$ (shown as the third tree in Figure \ref{ind}), of rank $k + 1$, with 1,1 root $y$ whose right child $x$ is also 1,1.  By the induction hypothesis, reinsertion of the deleted leaf into $C'$ does $k$ promotions and converts $C'$ into a tree $C''$ in $E$ of rank $k$.  In $T'$, the same reinsertion converts the right subtree of $T'$ into $C''$, making $x$ 0,1.  This causes $x$ and then $y$ to be promoted, for a total of $k + 2$ promotions, and produces the tree $T''$ in Figure \ref{ind}, which is a tree in $E$ of type $R$.
\end{proof}
\clearpage

\begin{figure}
\centering
\includegraphics[width=5in]{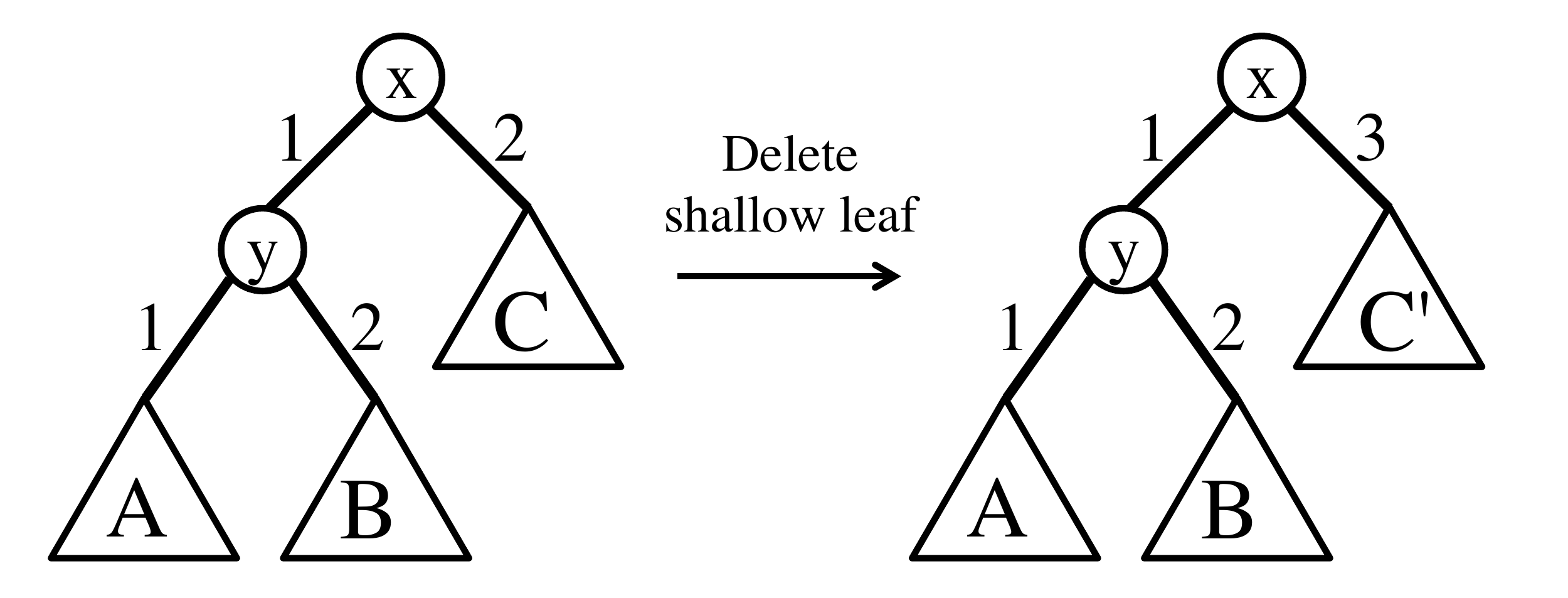}
\includegraphics[width=5in]{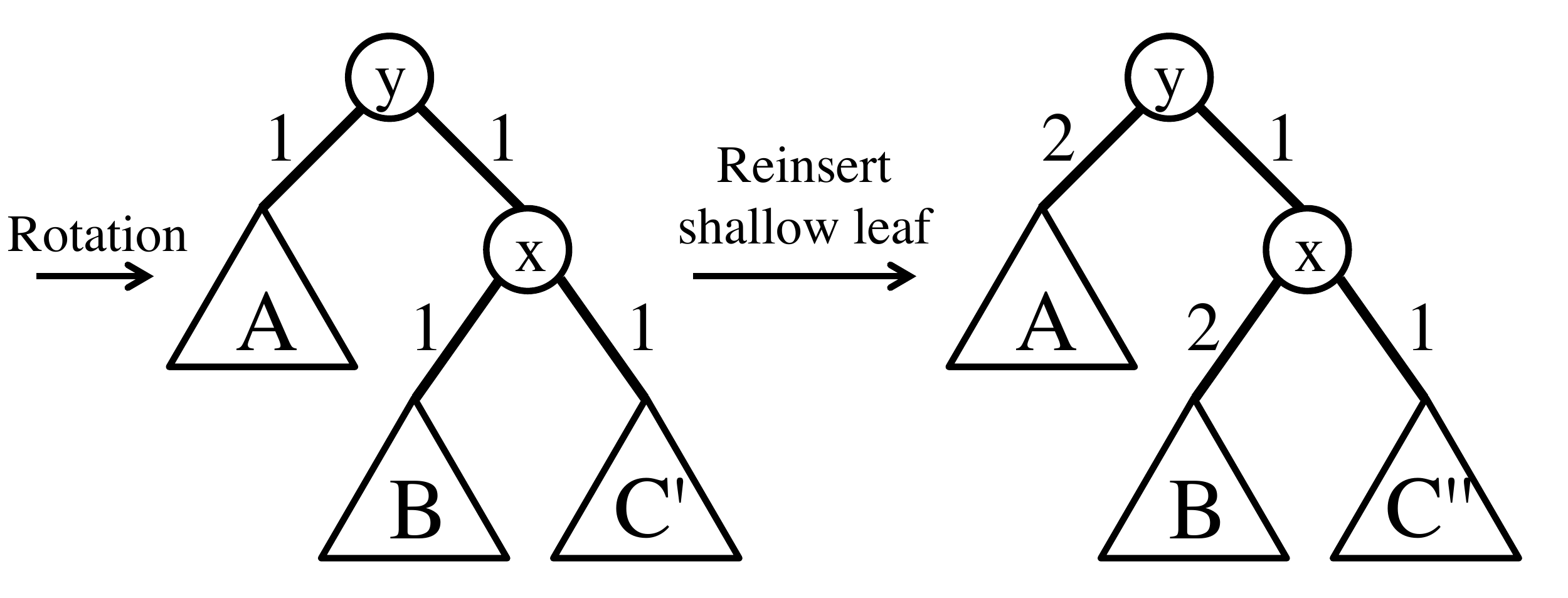}

\caption{Deletion and insertion of the shallow leaf in a type-$L$ tree of rank $k+2$}
\label{ind}
\end{figure}

\begin{rem}
The proof of Theorem \ref{thm:cost} implies that if one starts with a tree $T$ in $E$ of even rank $k$ and does $2^{k/2}$ deletion-reinsertion pairs, the final tree will be $T$.
\end{rem}

\begin{cor}
For infinitely many $n$, there is a sequence of $3n$ intermixed insertions and deletions on an initially empty AVL tree that takes $\Theta(n \log n)$ rotations.
\end{cor}

\begin{proof}
Let $T$ be any tree in $E$.  If $T$ has $n$ nodes, its height is $\Theta(\log n)$ since it is an AVL tree \cite{adelson}.  Apply Theorem \ref{thm:build} to build $T$ in $n$ insertions.  Then repeat the following pair of operations $n$ times: delete the shallow leaf; reinsert the deleted leaf.  By Theorem \ref{thm:cost}, the total number of rotations will be $\Theta(n \log n)$.
\end{proof}

\section*{Acknowledgments}
The third author thanks Uri Zwick for requesting verification of the claim in \cite{haeupler} that deletions in AVL trees have high amortized rotation cost, providing motivation to write this note.

The first author is partially supported by the Italian Ministry of Education, University, and Research (MIUR) under PRIN 2012C4E3KT national research project AMANDA --- Algorithmics for Massive and Networked Data.

%\section*{References}
\bibliographystyle{plain}
\bibliography{AVLbib}

%% The Appendices part is started with the command \appendix;
%% appendix sections are then done as normal sections
%% \appendix

%% \section{}
%% \label{}

\end{document}